\documentclass[12pt]{iopart}
\usepackage[all]{xy}
\usepackage{amsthm,amssymb,nath}
\usepackage{graphicx}

\def\Includegraphics[#1]#2{{\tt[#1]#2}}

\let\epsilon\varepsilon

\def\eqref#1{{\rm(\ref{#1})}}
\newtheorem{proposition}{Proposition}

\newtheorem{remark}{Remark}

\def\dif{\mathop{}\!\mathrm d}

\def\arctanh{\mathop{\mathrm{arctanh}}\nolimits}

\def\sn{\mathop{\mathrm{sn}}\nolimits}
\def\dn{\mathop{\mathrm{dn}}\nolimits}
\def\cn{\mathop{\mathrm{cn}}\nolimits}

\def\Scaling{\mathfrak S}
\def\xTranslation{\mathfrak T^{x}}
\def\yTranslation{\mathfrak T^{y}}

\usepackage{enumitem}

\begin{document}

\title{More exact solutions of the constant astigmatism equation}
\author{Adam Hlav\'a\v{c}}
\address{Mathematical Institute in Opava, Silesian University in
  Opava, Na Rybn\'\i\v{c}ku 1, 746 01 Opava, Czech Republic.
  {\it E-mail}: Adam.Hlavac@math.slu.cz}
\date{}


\begin{abstract}
By using B\"acklund transformation for the sine-Gordon equation, new periodic exact solutions of the constant astigmatism equation 
$ z_{yy} + ({1}/{z})_{xx} + 2 = 0 $ are generated from a~seed which corresponds to Lipschitz surfaces of constant astigmatism. 
\end{abstract}

\section{Introduction}
In this paper, we construct new exact solutions of the \it constant astigmatism equation (CAE)\rm 
$$\numbered\label{CAE}
z_{yy} + (\frac1z)_{xx} + 2 = 0,
$$
the Gauss equation for \it constant astigmatism surfaces \rm immersed in Euclidean space.
These surfaces are defined by the condition $\rho_2 - \rho_1 = \rm const \not = 0$, where $\rho_1$, $\rho_2$ are the principal
radii of curvature. 

It is well known \cite{Bia I, Bia II, Ku, vLil, Lip, Ri} that evolutes (focal surfaces) of constant astigmatism surfaces are pseudospherical, i.e. with constant negative Gaussian curvature. Conversely, involutes of pseudospherical surfaces corresponding to parabolic geodesic net are of constant astigmatism.  

After a century of oblivion, constant astigmatism surfaces 
reemerged in 2009 in the work \cite{B-M I} concerning the systematic search for integrable classes of Weingarten surfaces. 
In the paper, the surfaces were given a name and Equation \eqref{CAE} was derived for the first time. 
The geometric connection between surfaces of constant astigmatism and pseudospherical surfaces provides
 a transformation \cite[Sect.~6 and 7]{B-M I} between the CAE and the famous \it sine-Gordon equation \rm
$$\numbered\label{SG}
\omega_{\xi\eta} =  \sin \omega,
$$
which is the Gauss equation of pseudospherical surfaces parameterized by Chebyshev--asymptotic coordinates. 

To our best knowledge, the list of known exact solutions of the CAE is exhausted by: 
\begin{itemize}
\item Von Lilienthal solutions $z = c-y^2 $ and 
$z = 1/(c - x^2)$, where $c$ is a constant. 
For details and corresponding surfaces of constant astigmatism see \cite{B-M I, Hl, vLil}.
\item Lipschitz solutions \cite[Thm. 1]{H-M II} corresponding to Lipschitz surfaces of constant astigmatism \cite{Lip}, 
see also \cite{Hl} for pictures of the surfaces.
\item Solutions parameterized by an arbitrary function of a single variable obtained by Manganaro and Pavlov \cite{M-P}. They have no
counterpart solutions of the sine-Gordon equation.
\item Examples of exact solutions generated by 
using \it reciprocal transformations \rm \cite{H-M III}. 
\item Multisoliton solutions \cite{Hl} defined as counterparts of multisoliton solutions of the sine-Gordon equation. They arise from von Lilienthal seed (i.e. from zero solution of the sine-Gordon equation) using nonlinear superposition formula 
for the CAE (see Proposition \ref{prop1} below). 
\end{itemize}  

The purpose of this paper is to construct another seed solution of the CAE which is ready to be transformed to infinite number of new solutions by nonlinear superposition formula \eqref{prop1eq}. In order to do it, we need a \it nonzero \rm sine-Gordon seed and its 
B\"acklund transformation. Techniques for obtaining such pairs can be found e.g. in \cite{Ho-Mi} or \cite{Gr}; 
the simplest initial sine-Gordon solution being the so called ``travelling wave''. 
The corresponding solutions of the CAE are the Lipschitz solutions 
\cite{H-M II, Lip}.  

 In other words, in this paper, we are going to reproduce Lipschitz solutions to obtain an infinite set of new solutions from it. 
However, the explicit form of the general Lipschitz solution as introduced in \cite[Thm.~1]{H-M II}, 
is not suitable for our approach; we seek the seed solution parameterized by coordinates $\xi,\eta$ in order to be prepared for the formula \eqref{prop1eq} below.

Note that the class of Lipschitz solutions exactly matches the class of solutions invariant with respect to
the Lie symmetries of the CAE. The exhaustive list of them consists of 
translations $\xTranslation : (x,y,z) \mapsto (x + t, y, z)$, 
$\yTranslation : (x,y,z) \mapsto (x, y + t, z)$ and the scaling
$\Scaling : (x,y,z) \mapsto (`e^{-t} x, `e^{t} y, `e^{2t} z)$, where $t$ is a
real parameter.

\section{Construction of a seed solution} 
The \it B\"acklund transformation \rm \cite{Bae}, see also \cite{Bia I, Bia G}, 
for the sine-Gordon equation \eqref{SG} takes a~solution $\omega$ 
and produces a new solution $\omega^{(\lambda)}$,  
given by the system
$$\numbered\label{BT}
(\frac{\omega^{(\lambda)} - \omega}{2})_\xi  =   \lambda\sin(\frac{\omega^{(\lambda)} + \omega}{2}), \qquad
(\frac{\omega^{(\lambda)} + \omega}{2})_\eta =   \frac{1}{\lambda}\sin(\frac{\omega^{(\lambda)} - \omega}{2}),
$$
$\lambda$ being called a \it B\"acklund parameter. \rm
What is more, the superposition formula \cite{Bia b, Bia G}
$$\numbered\label{supSG}
\tan\frac{\omega^{(\lambda_1\lambda_2)} - \omega}{4} = \frac{\lambda_1 + \lambda_2}{\lambda_1 - \lambda_2}\tan\frac{\omega^{(\lambda_1)} - \omega^{(\lambda_2)}}{4}
$$ 
allows us to obtain the solution $\omega^{(\lambda_1\lambda_2)}$, the B\"acklund transformation of $\omega^{(\lambda_1)}$
with B\"acklund parameter $\lambda_2$, by purely algebraic manipulations.

In \cite{H-M II} we showed that solutions of the sine-Gordon equation 
corresponding to Lipchitz's solutions of the CAE satisfy 
$$\numbered\label{LipCond}
\omega_{\xi} = k\omega_\eta,
$$ 
where $k$ is a nonzero constant. Thus, they are of the form $\omega(k\xi + \eta + C) $, where $C$~is a~constant. 
Let us perform a transformation to new coordinates $\alpha = k\xi + \eta$ and $\beta = k\xi - \eta$ in which  
the sine-Gordon equation turns out to be
$$\numbered \label{sGab} 
\omega_{\alpha\alpha} - \omega_{\beta\beta} = \frac{1}{k}\sin \omega.
$$ 
The condition \eqref{LipCond} reduces to 
$$
\omega_\beta = 0
$$
and, therefore, the seed sine-Gordon solutions we are working with depend solely on $\alpha$~and they satisfy the ODE
$$\numbered\label{ODE}
k\omega_{\alpha\alpha} = \sin \omega.
$$

Let $\omega$ be a solution of \eqref{ODE}. Let us proceed to its B\"acklund transformation $\omega^{(\lambda)}$ 
in accordance with \cite{Ho-Mi}. 
Firstly, $\omega^{(\lambda)}$ can be conveniently written  as 
$$\numbered\label{BTomega}
\omega^{(\lambda)} = 4\arctan \delta^{(\lambda)} - \omega, 
$$
where $\delta^{(\lambda)}$ satisfies the system
$$\numbered\label{delta}
\delta^{(\lambda)}_\alpha = 
\frac{({\delta^{(\lambda)}}^2 - 1) \sin \omega + 2 \delta^{(\lambda)} (\cos \omega + \frac{\lambda^2}{k}) + \lambda ({\delta^{(\lambda)}}^2 + 1) \omega_\alpha }{4\lambda}, 
\\
\delta^{(\lambda)}_\beta = 
\frac{(-{\delta^{(\lambda)}}^2 + 1) \sin \omega - 2 \delta^{(\lambda)} (\cos \omega - \frac{\lambda^2}{k}) + \lambda ({\delta^{(\lambda)}}^2 + 1) \omega_\alpha }{4\lambda}.
$$
The second equation in \eqref{delta} is separable and it has a solution
$$\numbered\label{deltasol}
\delta^{(\lambda)} = \frac{f}{a k^2} + \frac{c}{a k} \tanh [c(\beta+b(\alpha)+K)],
$$
where
$$\numbered\label{afc}
a = \frac{\sin\omega}{4k\lambda} - \frac{\omega_\alpha}{4k}, \qquad
f = \frac{\lambda}{4}-\frac{k \cos \omega}{4\lambda}, \\
c = \frac{\sqrt{\lambda^4 -2 k \lambda^2\cos \omega  - k^2 (\lambda^2 \omega_\alpha^2-1) }}{4k \lambda}
$$
and $b(\alpha)$ is yet unknown function of $\alpha$. 
Substituting \eqref{deltasol} into the first equation in \eqref{delta} one obtains equation for $b$, namely
$$\numbered\label{b}
b_\alpha = \frac{\lambda \omega_{\alpha} + \sin\omega}{\lambda \omega_{\alpha} - \sin\omega}.
$$

Furthermore, it is possible to express $\omega_\alpha$ in terms of $\omega$. 
Multiplying both sides of~\eqref{ODE} by $\omega_\alpha$ and integrating, one can reduce the order of the equation which becomes 
$$\label{omegaalpha}
k\omega_\alpha^2  = -2\cos \omega + 2l, 
$$
$l$ being a constant. 
Solving for $\omega_\alpha$ we obtain
$$\numbered\label{omegaalpha2}
\omega_{\alpha} = \pm\sqrt{\frac{2l - 2\cos\omega}{k}}.
$$
Substituting into equation for $c$ in \eqref{afc} one can easily see that
$$
c = \frac{\sqrt{\lambda^4-2kl\lambda^2+k^2}}{4k\lambda} = \rm const.
$$

In \cite{H-M I} we introduced the method of finding a solution of the CAE corresponding to given sine-Gordon solution $\omega$ and its B\"acklund transformation $\omega^{(\lambda)}$. 
The solution in question is given in terms of
 \it associated potentials \rm $g^{(\lambda)},x^{(\lambda)},y^{(\lambda)}$ \cite{H-M I, Hl} defined as solutions of the system
$$\numbered\label{aspot0}
x^{(\lambda)}_\xi = \lambda g^{(\lambda)} \sin\frac{\omega^{(\lambda)} + \omega}{2},
\qquad
x^{(\lambda)}_\eta = \frac1\lambda g^{(\lambda)} \sin\frac{\omega^{(\lambda)} - \omega}{2},
\\
y^{(\lambda)}_\xi = \frac{\lambda}{g^{(\lambda)}}  \sin\frac{\omega^{(\lambda)} + \omega}{2} ,
\qquad
y^{(\lambda)}_\eta = - \frac{1}{\lambda g^{(\lambda)}} \sin\frac{\omega^{(\lambda)} - \omega}{2} ,
\\ 
g^{(\lambda)}_\xi = g^{(\lambda)}\lambda \cos\frac{\omega^{(\lambda)} + \omega}{2},
\qquad 
g^{(\lambda)}_\eta = g^{(\lambda)}\frac1\lambda \cos\frac{\omega^{(\lambda)} - \omega}{2}.
$$
Expressing $z^{(\lambda)} = 1/{g^{(\lambda)}}^{2}$ in terms of $x^{(\lambda)},y^{(\lambda)}$, one obtains a solution $z(x,y)$ of the CAE.

Moreover, according to \cite{Bia a}, considering a dependence of $\omega^{(\lambda)}(\xi,\eta)$ on an integration constant~$K$, potentials $x^{(\lambda)}(\xi,\eta)$ and $g^{(\lambda)}(\xi,\eta)$ can be obtained by algebraic handling and differentiation with respect to $K$, namely  
$$\numbered\label{difK1}
g^{(\lambda)} = \frac{\dif \omega^{(\lambda)}}{\dif K}, \qquad
x^{(\lambda)} = -2\frac{\dif \ln g^{(\lambda)}}{\dif K}.
$$ 
Nevertheless, computing $y^{(\lambda)}(\xi,\eta)$ requires integration of the second pair of equations in the system \eqref{aspot0}.

Let us slightly change the notation. Let $\omega^{[0]} = \bar{\omega}^{[0]} $ be a solution of the sine-Gordon equation. 
Fix B\"acklund parameters $\lambda_1, \dots, \lambda_{k + 1}$ and, according to the diagram 
\begin{displaymath} \numbered \label{lattice}
    \xymatrix{ 
\omega^{[0]} \ar[r]^{\lambda_2} \ar[d]_{\lambda_1} 	& \bar{\omega}^{[1]} \ar[r]^{\lambda_3} \ar[d]_{\lambda_1} & \bar{\omega}^{[2]} \ar[d]_{\lambda_1} \ar[r]^{\lambda_4}  & \bar{\omega}^{[3]} \ar[d]_{\lambda_1} \ar[r]^{\lambda_5}  & \bar{\omega}^{[4]} \ar[d]_{\lambda_1} \ar[r]^{\lambda_6} & \dots
\\
\omega^{[1]} \ar[r]^{\lambda_2}  & 	\omega^{[2]}  \ar[r]^{\lambda_3} & 	\omega^{[3]} \ar[r]^{\lambda_4} & 	\omega^{[4]} \ar[r]^{\lambda_5} & 	\omega^{[5]} \ar[r]^{\lambda_6} & \dots,
}
\end{displaymath}
denote
$$\numbered\label{defk}
\omega^{[k]} = \omega^{(\lambda_1\lambda_2\ldots\lambda_k)}, \qquad
\bar{\omega}^{[k]} = \omega^{(\lambda_2\lambda_3\ldots\lambda_{k+1})}.
$$

Let $g^{[j]}, x^{[j]}, y^{[j]}$ be  associated potentials corresponding to the pair $\bar{\omega}^{[j-1]}, \omega^{[j]}$. 
They satisfy (see \cite{Hl} for details) recurrences
$$
x^{[j+1]}
 = \frac{\lambda_{j+1} \lambda_1}{\lambda_{j+1}^2 - \lambda_1^2} 
   (x^{[j]}
     -  
      \frac{2 \lambda_{j+1} \lambda_1  \sin\frac{\bar{\omega}^{[j]} - \omega^{[j]}}{2}}
         {\lambda_{j+1}^2 + \lambda_1^2
          - 2 \lambda_{j+1} \lambda_1 \cos\frac{\bar{\omega}^{[j]} - \omega^{[j]}}{2}
          } g^{[j]} ),
\\
y^{[j+1]}
 = \frac{\lambda_{j+1}^2 - \lambda_1^2}{\lambda_{j+1}\lambda_1} y^{[j]}
   - \frac{2}{g^{[j]}} \sin\frac{\bar{\omega}^{[j]} - \omega^{[j]}}{2},
\\
g^{[j+1]} = 
\frac{-\lambda_{j+1}\lambda_1}{\lambda_{j+1}^2+\lambda_1^2-2\lambda_{j+1}\lambda_1\cos\frac{\bar{\omega}^{[j]} - \omega^{[j]}}{2}}g^{[j]},
$$
which are easy to solve as follows.

\begin{proposition}\cite{Hl}\label{prop1}
Let $x_1, y_1, g_1$ be the associated potentials corresponding to the pair $\omega^{[0]}, \omega^{[1]}$ of sine-Gordon solutions.
Let $S^{[j]}$ be $4\times 4$ matrices with entries defined by formulas 
$$
S^{[j]}_{11} = \frac{\lambda_{j+1} \lambda_1}{\lambda_{j+1}^2 - \lambda_1^2}, \quad
S^{[j]}_{13} = -\frac{\lambda_{j+1}^2 \lambda_1^2}{\lambda_{j + 1}^2 - \lambda_1^2} \times \frac{2  \sin\frac{\bar{\omega}^{[j]} - \omega^{[j]}}{2} }
{\lambda_{j+1}^2 + \lambda_1^2 - 2 \lambda_{j+1} \lambda_1 \cos\frac{\bar{\omega}^{[j]} - \omega^{[j]}}{2} }, \\
S^{[j]}_{22} = \frac{\lambda_{j+1}^2 - \lambda_1^2}{\lambda_{j+1}\lambda_1}, \quad
S^{[j]}_{24} = -2  \sin\frac{\bar{\omega}^{[j]} - \omega^{[j]}}{2}, \\
S^{[j]}_{33} = \frac{1}{S^{[j]}_{44}} = \frac{-\lambda_{j+1}\lambda_1}{\lambda_{j+1}^2+\lambda_1^2-2\lambda_{j+1}\lambda_1\cos\frac{\bar{\omega}^{[j]} - \omega^{[j]}}{2} 
} 
$$
all the other entries being zero. Let 
$$\numbered\label{prop1eq}
(
\begin{array}{c}
x^{[n]} \\
y^{[n]} \\
g^{[n]} \\
1/g^{[n]} 
\end{array}
)
= (
\prod_{i = 1}^{n - 1} S^{[i]} )
(
\begin{array}{c}
x_1 \\
y_1 \\
g_1 \\
1/g_1 
\end{array}
).
$$
Then $x^{[n]}, y^{[n]}, g^{[n]}$ are the associated potentials corresponding to the pair $\bar{\omega}^{[n-1]}, \omega^{[n]}$. Moreover, 
if $z^{[n]} = 1/{g^{[n]}}^2 $, then $z^{[n]}(x^{[n]},y^{[n]})$ is a solution of the constant astigmatism equation \eqref{CAE}.
\end{proposition}

Our main task is to find a solution of the CAE corresponding to $\omega$ (a solution of \eqref{ODE}), 
and $\omega^{(\lambda)}$ given by \eqref{BTomega}. 
For this pair, we need to solve the system \eqref{aspot0}, which in terms of  $\alpha,\beta$ turns out to be 
$$\numbered\label{aspot}
x^{(\lambda)}_\alpha = g^{(\lambda)} \times \frac{  2 \delta^{(\lambda)} (\lambda^2 + k\cos \omega) + k({\delta^{(\lambda)}}^2-1) \sin \omega  }
{2k\lambda ({\delta^{(\lambda)}}^2+1)}, \\
x^{(\lambda)}_\beta = g^{(\lambda)} \times \frac{ 2 \delta^{(\lambda)} (\lambda^2 - k\cos \omega )-k({\delta^{(\lambda)}}^2-1)\sin \omega  }
{2k\lambda ({\delta^{(\lambda)}}^2+1)},\\
y^{(\lambda)}_\alpha = \frac{1}{g^{(\lambda)}} \times \frac{  2 \delta^{(\lambda)} (\lambda^2 - k\cos \omega)-k({\delta^{(\lambda)}}^2-1)\sin \omega  }{2k\lambda ({\delta^{(\lambda)}}^2+1)}, \\
y^{(\lambda)}_\beta = \frac{1}{g^{(\lambda)}} \times \frac{  2 \delta^{(\lambda)} (\lambda^2 + k\cos \omega) + k({\delta^{(\lambda)}}^2-1)\sin \omega }{2k\lambda ({\delta^{(\lambda)}}^2+1)},
\\
g^{(\lambda)}_\alpha = g^{(\lambda)} \times
\frac{(1-{\delta^{(\lambda)}}^2) (\lambda^2 + k\cos \omega) + 2 k\delta^{(\lambda)} \sin \omega }
{2 k\lambda ({\delta^{(\lambda)}}^2+1)}, \\
g^{(\lambda)}_\beta = g^{(\lambda)} \times
\frac{(1-{\delta^{(\lambda)}}^2) (\lambda^2-k\cos \omega) - 2 k\delta^{(\lambda)} \sin \omega }
{2k \lambda ({\delta^{(\lambda)}}^2+1)}.
$$ 

\begin{proposition} 
Let $\omega$ be a solution of \eqref{ODE} and let $\omega^{(\lambda)} = 4\arctan \delta^{(\lambda)} - \omega$ 
be its B\"acklund 
transformation with parameter $\lambda$, where $\delta^{(\lambda)}$ is given by \eqref{deltasol}. Let $a,f,c$ be defined by \eqref{afc} and let $b$ satisfy \eqref{b}. 
Then the associated potentials $x^{(\lambda)},y^{(\lambda)},g^{(\lambda)}$ 
corresponding to the pair $\omega, \omega^{(\lambda)}$ 
are given by formulas
$$\numbered\label{gx}
x^{(\lambda)} = \frac{8 c^2 k f \cosh 2 B+4 c (f^2+k^4 a^2+c^2 k^2) \sinh 2 B}{(f^2+k^4 a^2+c^2 k^2) \cosh 2 B+2  c k f \sinh 2 B+f^2+k^4 a^2-c^2 k^2}, \\
y^{(\lambda)} = 
(\frac{f\sin\omega}
{16\lambda c^2k^2a}
-\frac {2 f-\lambda}{8c^2k})
\cosh 2B \\  
-(\frac {  ( k^4 a^2-c^2k^2-f^2)\sin\omega}
{32\lambda c^3k^3a}
+
\frac {f(2f-\lambda)}{8c^3k^2}) 
\sinh 2B 
 \\
+\frac{1}{\lambda^4-2kl\lambda^2+k^2}(\frac{\lambda^4-k^2}{2}
\beta 
+ \frac{\lambda^4+k^2}{2} \alpha
- {k \lambda^2} \int\cos\omega \dif\alpha), \\
g^{(\lambda)} = \frac{4c^2 k^3 a (1-\tanh^2  B) }{k^4 a^2+f^2+2 c k f\tanh  B + c^2 k^2\tanh^2  B }, 
$$
where $B = c(\beta+b+K)$ and $l$ is a constant. Moreover, 
if $z^{(\lambda)} = 1/{g^{(\lambda)}}^{2} $, then $z^{(\lambda)}(x^{(\lambda)},y^{(\lambda)})$ is a solution of the constant astigmatism equation \eqref{CAE}.
\end{proposition}

\begin{proof}
According to \cite[Prop. 4]{H-M I} 
(see also \cite{Bia a, Bia II}) 
$g^{(\lambda)}$ and $x^{(\lambda)}$ can be obtained without integration. 
One can simply consider the dependence of 
$$
\omega^{(\lambda)} = 4\arctan \delta^{(\lambda)} - \omega
$$ 
on the integration constant $K$ and differentiate with respect to it, namely 
$$\numbered\label{difK}
g^{(\lambda)} = \frac{\dif \omega^{(\lambda)}}{\dif K} , \qquad
x^{(\lambda)} = -2\frac{\dif \ln g^{(\lambda)}}{\dif K}.
$$
Performing \eqref{difK} one immediately obtains formulas for $x^{(\lambda)}$ and $g^{(\lambda)}$.

A variable $y^{(\lambda)}$ can be obtained by integrating the middle two equations in \eqref{aspot}. The latter can be routinely 
integrated with respect to $\beta$,   
$$\numbered\label{ylambda}
y^{(\lambda)} = 
(\frac{f\sin\omega}
{16\lambda c^2k^2a}
-\frac {2 f-\lambda}{8c^2k})
\cosh 2c(\beta+b+K) \\  
-(\frac {  ( k^4 a^2-c^2k^2-f^2)\sin\omega}
{32\lambda c^3k^3a}
+
\frac {f(2f-\lambda)}{8k^2c^3}) 
\sinh 2c(\beta+b+K) 
 \\
+\frac12 \frac{\lambda^4-k^2}{\lambda^4-2kl\lambda^2+k^2}
\beta
+F(\alpha),
$$
where $F(\alpha)$ is yet unknown function of $\alpha$. Substituting \eqref{ylambda} into the equation for $y^{(\lambda)}_\alpha$ one obtains 
$$
F_\alpha = \frac12  
\frac{\lambda^4-2 k\lambda^2 \cos\omega + k^2}
{\lambda^4-2kl\lambda^2+k^2}
$$
meaning that
$$
F = \frac12 \frac{\lambda^4+k^2}{\lambda^4-2kl\lambda^2+k^2} \alpha
- \frac{k \lambda^2}{\lambda^4-2kl\lambda^2+k^2} \int\cos\omega \dif\alpha.
$$
\end{proof}

Once a seed solution $(x^{(\lambda)},y^{(\lambda)},g^{(\lambda)})$ is known, it is a matter of routine to generate arbitrary number of new solutions using Proposition \ref{prop1}.

\begin{remark} \label{rem1}\rm
According to \cite{Dav}, we have the following results of integration of \eqref{omegaalpha2}, 
\begin{itemize}
\item for $l/k>1/k$ (case A)
$$\numbered\label{sol1}
\omega_0 = 2\arccos [\sn (\frac{-\alpha p}{\sqrt{k}} ; \frac1p)], 
$$
\item for $|l/k|<1/k$  (case B)
$$\numbered\label{sol2}
\omega_0 = 2\arcsin [\dn (\frac{\alpha}{\sqrt{k}} ; p)],  
$$
\item for $l = 1$ (case C)
$$\numbered\label{sol3}
\omega_0 = 4\arctan (\exp{\frac{\alpha}{\sqrt{k}}}),  
$$
\end{itemize}
where we have denoted
$$\numbered\label{p}
p = \sqrt{\frac{1 + l}{2}}.
$$
Note that case C coincides with \it one-soliton solution \rm of the sine-Gordon equation, the B\"acklund transformation of zero solution $\omega = 0$. For details, corresponding solutions of the CAE and surfaces of constant astigmatism we refer to \cite{Hl}. Hence, the case when $l = 1$ is taken out of consideration in the sequel.
\end{remark}

\begin{remark} \rm
Substituting solutions from Remark \ref{rem1} into \eqref{b} one can routinely integrate to obtain:
\begin{itemize}
\item in case A (for $l/k>1/k$)
$$\numbered\label{b2}
b = -\alpha + \frac{2k\lambda}{\sqrt{(k+\lambda^2)^2-4kp^2\lambda^2}} 
\arctanh \frac{2 k\sn^2( \frac{\alpha p}{\sqrt{k}}, \frac1p) - \lambda^2- k}
{\sqrt{(k+\lambda^2)^2-4kp^2\lambda^2}} \\
+ 
\frac{\lambda^2\sqrt{k}}{p}  \sum_{\gamma} 
\frac{\gamma^2 -p^2 }{\gamma^2(2\gamma^2 k-\lambda^2-k)} \{\Pi[\sn(\frac{\alpha p}{\sqrt{k}},\frac1p), \frac{1}{\gamma^2}, \frac1p] \\
+\frac{\gamma p}{2\sqrt{(1-\gamma^2)(p^2-\gamma^2)}} 
(\ln [
\gamma^2 p^2 (-\gamma^2+\sn^2(   \frac{\alpha p}{\sqrt{k}}, \frac1p))]  \\
- 
\ln [(p^2-2 \gamma^2+1) \sn^2(  \frac{\alpha p}{\sqrt{k}}, \frac1p)
-2 p^2+\gamma^2+\gamma^2 p^2\\
-2  p \sqrt{(1-\gamma^2)(p^2-\gamma^2)} \cn(   \frac{\alpha p}{\sqrt{k}}, \frac1p) \dn(   \frac{\alpha p}{\sqrt{k}}, \frac1p) 
])
\} 
,
$$
where $\gamma$ are roots of $ k\gamma^4-(k+\lambda^2)\gamma^2+\lambda^2 p^2$,  and
\item in case B (for $|l/k|<1/k$)
$$\numbered\label{b3}
b = -\alpha - 
\frac{2k\lambda}{\sqrt{4kp^2\lambda^2-(k+\lambda^2)^2}}
\arctan \frac{2kp^2\sn^2(\frac{\alpha}{\sqrt{k}}, p) -\lambda^2 -k}{\sqrt{4kp^2\lambda^2-(k+\lambda^2)^2}} \\
-\sqrt{k}\lambda^2\sum_{\gamma} 
\frac{1-\gamma^2}{\gamma^2(2\gamma^2kp^2-\lambda^2-k)}
\{\Pi[\sn(\frac{\alpha}{\sqrt{k}}, p), \frac{1}{\gamma^2}, p] \\
+\frac{ \gamma p}{ 2\sqrt{(1-\gamma^2)(p^2-\gamma^2)}}
(
\ln[\gamma^2 p^2 (\sn^2(\frac{\alpha}{\sqrt{k}}, \frac{1}{p})-\gamma^2)]  \\
- 
\ln[(p^2-2 \gamma^2+1) 
\sn^2(\frac{\alpha}{\sqrt{k}}, \frac{1}{p})-2 p^2+\gamma^2+\gamma^2 p^2 \\
-2  p  \sqrt{(1-\gamma^2)(p^2-\gamma^2)} 
\cn(\frac{\alpha}{\sqrt{k}}, \frac{1}{p}) 
\dn(\frac{\alpha}{\sqrt{k}}, \frac{1}{p}) 
]
)\},
$$
where $\gamma$ are roots of $k p^2 \gamma^4 - (k+\lambda^2)\gamma^2+\lambda^2$. 
Here $\Pi$ denotes incomplete elliptic integral of the third kind.
\end{itemize}
\end{remark}

\begin{remark} \label{rem0}\rm
A closer look shows that the solution \eqref{gx} is periodic in the $y$-direction. Indeed, the shifts 
$\alpha\mapsto\alpha + 4\sqrt{k}`K(p)$ (in case A) and $\alpha\mapsto\alpha + 2\sqrt{k}`K(p)$ (in case B) leave 
$x^{(\lambda)}$ and $g^{(\lambda)}$ unchanged, while the $y^{(\lambda)}$ is translated by
$$
P_{\rm A} = \mathfrak{Re} (\frac{2 \sqrt{k} [(k + \lambda^2)^2 -4 k \lambda^2 p^2 ] `K(p)}{
(k + \lambda^2)^2 - 4 k^2\lambda^2 } 
+ \frac{8 k^{\frac32} \lambda^2 p `E[\sn(p`K(p) , \frac{1}{p}), \frac{1}{p}] }{(k + \lambda^2)^2 - 4 k^2\lambda^2 })
$$
in case A and by
$$
P_{\rm B} = \frac{2\sqrt{k} (k-\lambda^2)^2  }{(k + \lambda^2)^2 -4 k^2 \lambda^2 } `K(p)
+
\frac{8 k^{\frac32} \lambda^2 }{(k + \lambda^2)^2 -4 k^2 \lambda^2} `E(p)
$$
in case B. Here $`K$ and $`E$ denote complete elliptic integrals of the first and the second kind respectively.
Remarkably enough, the shift 
$\beta\mapsto\beta + 2\pi`i$ also leaves 
$x^{(\lambda)}$ and $g^{(\lambda)}$ unchanged and translates  $y^{(\lambda)}$ by complex period
$$
\frac{-2k\lambda\pi `i (k^2-\lambda^4)}{(\lambda^4-2 k l \lambda^2 +k^2)^{3/2}}.
$$
Finally, a detailed look at the formula \eqref{prop1eq} reveals that the $n$-th solution, arising from the periodic seed \eqref{gx}, is also periodic with period 
$
P_{\rm A} \times (\prod_{i = 1}^{n - 1} S^{[i]}_{22} )
$
in case A and 
$
P_{\rm B} \times(\prod_{i = 1}^{n - 1} S^{[i]}_{22} )
$
in case B.
\end{remark}

A graph of the solution $z(x,y) = z^{(\lambda)}(x^{(\lambda)},y^{(\lambda)})$, where $x,y$ and $z = 1/g^2$ are given by \eqref{gx} 
and are parameterized by coordinates $\xi,\eta$, can be easily plotted for both cases A and B, see Figure \ref{CAsol1}. Periodicity in $y$-direction can be clearly seen.
\begin{figure}[ht] 
\begin{center} 
\includegraphics[scale=0.36]{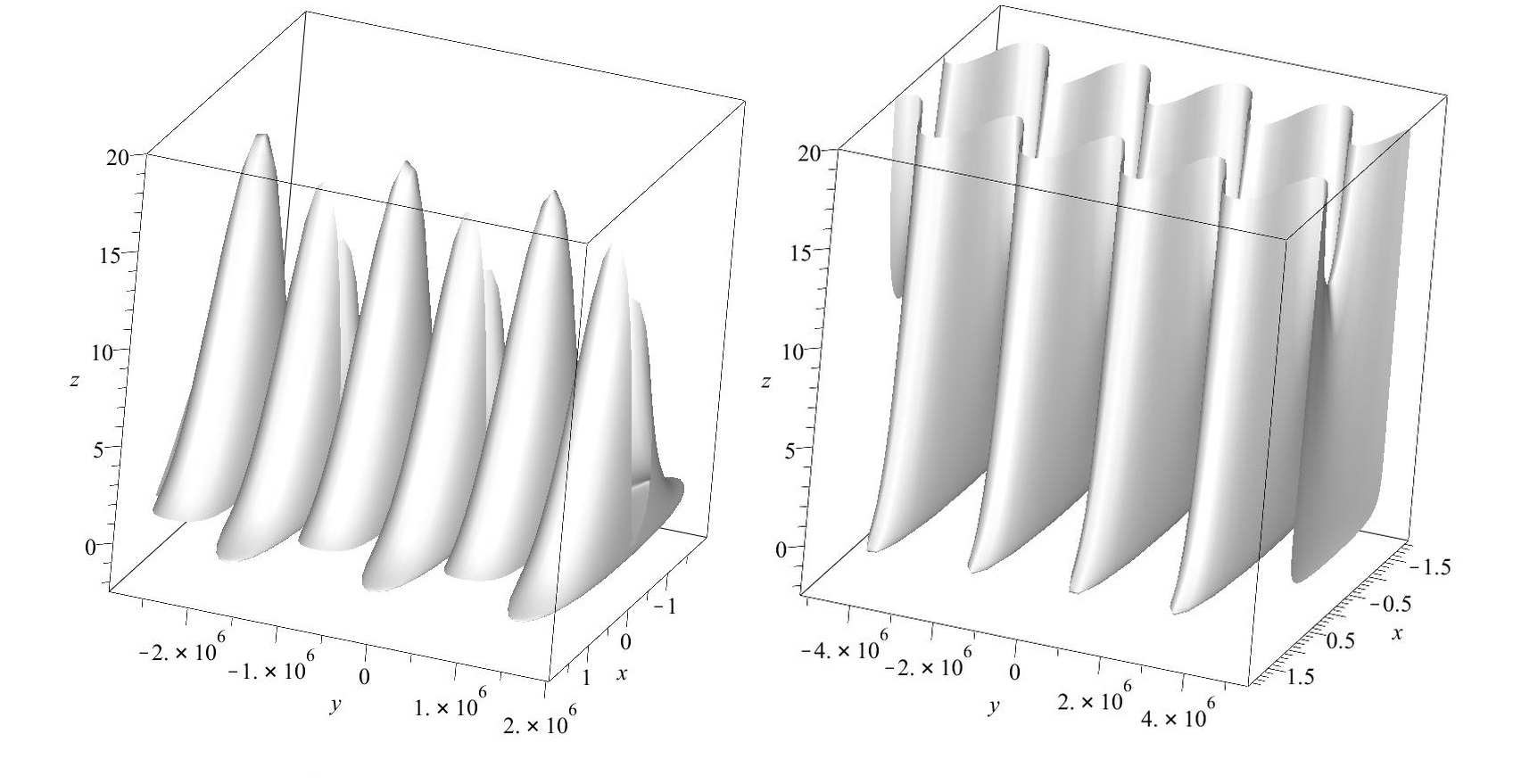}
\caption{Solutions $z^{(\lambda)}(x^{(\lambda)},y^{(\lambda)})$ of the CAE, $\lambda = 1.001$, $k = 1$, $K = 0$, $l = 3/2$ (case A, left),  $l = 1/2$ (case B, right).  } \label{CAsol1} 
\end{center}
\end{figure} 

\section{Surfaces of constant astigmatism}
Let $\mathbf r $ be a pseudospherical surface corresponding to a sine-Gordon solution $\omega(\xi,\eta)$, i.e. 
$\mathbf r (\xi,\eta)$, its unit normal $\mathbf n (\xi,\eta)$ and $\omega(\xi,\eta)$ satisfy the Gauss-Weingarten system  
$$
\mathbf r_{\xi\xi} = \omega_\xi (\cot \omega)  \mathbf r_\xi
-  \omega_\xi (\csc \omega) \mathbf r_\eta, \\
\mathbf r_{\xi\eta} = (\sin \omega) \mathbf n, \\
\mathbf r_{\eta\eta} = \omega_\eta (\cot \omega)  \mathbf r_\eta
 -  \omega_\eta (\csc \omega) \mathbf r_\xi, \\
\mathbf n_\xi = (\cot \omega) \mathbf r_\xi
 - (\csc \omega)  \mathbf r_\eta, \\
\mathbf n_\eta = (\cot\omega) \mathbf r_\eta
 - (\csc \omega) \mathbf r_\xi, 
$$
compatible by virtue of the sine-Gordon equation \eqref{SG}.

A B\"acklund transformation $\mathbf r^{(\lambda)}$ of the surface $\mathbf r$ is (see e.g. \cite{R-S})   
$$\numbered\label{BTsurf}
\mathbf r^{(\lambda)} = \mathbf r
 + \frac{2\lambda\csc\omega}{1 + \lambda^2}[\sin(\frac{\omega - \omega^{(\lambda)}}{2}) {\mathbf  r_{\xi}}
 + \sin(\frac{\omega + \omega^{(\lambda)}}{2}) {\mathbf r_{\eta}}],
$$
where $\omega^{(\lambda)}$ satisfies \eqref{BT}.
Substituting $\lambda = 1$ into \eqref{BTsurf} one obtains  a \it complementary \rm pseudospherical surface
$$\numbered\label{BTcomp}
\mathbf r^{(1)} = \mathbf r
 + [\sin(\frac{\omega - \omega^{(1)}}{2}) {\mathbf  r_{\xi}}
 + \sin(\frac{\omega + \omega^{(1)}}{2}) {\mathbf r_{\eta}}]\csc\omega.
$$
According to \cite{H-M I}, the unit normal, $\tilde{\mathbf n}^{(1)}$, of the corresponding constant astigmatism surface is 
$$
\tilde{\mathbf n}^{(1)} = \mathbf r^{(1)} - \mathbf r_0
$$
and a family of constant astigmatism surfaces $\tilde{\mathbf r}^{(1)}$, having evolutes $\mathbf r^{(1)}$ and $\mathbf r_0$, is given by 
$$\numbered\label{CAsurf}
\tilde{\mathbf{r}}^{(1)} = \mathbf{r}_0 -  \tilde{\mathbf{n}}^{(1)} \ln g^{(1)} + s\tilde{\mathbf{n}}^{(1)},
$$
where $g^{(1)} = g^{(\lambda)}|_{\lambda = 1} $ and $s$ is a real offsetting parameter, change of which moves every point along the normal, i.e. it takes the surface to a parallel one. 

Pseudospherical surfaces corresponding to the solutions satisfying \eqref{LipCond} were constructed by 
Zadadaev \cite{Zad} and, according to \cite{Ov}, they coincide with surfaces studied in the 19th century by Minding \cite{Mi}, see also \cite{Bia G, No}. In Zadadayev's parameterization by asymptotic coordinates the surfaces are given by formula
$$\numbered\label{PSZad}
\mathbf{r}_0 = 
\frac{\sqrt{k}}{p(k+1) }
(
\begin{array}{c}
2\sin\frac{\omega_0}{2} \sin[p (\xi - \eta)] \\
2\sin\frac{\omega_0}{2} \cos[p (\xi - \eta)] \\
\xi + \eta +\frac{1}{\sqrt{k}}\int\cos \omega_0\dif \alpha
\end{array}
),
$$
where $\omega_0 = \omega_0(\alpha) = \omega_0(k\xi + \eta)$ is one of the solutions \eqref{sol1}--\eqref{sol3} and the constants $k,p$ have the same meaning as in the previous section, see \eqref{LipCond} and \eqref{p}. For pictures corresponding to all three cases, A, B and C, see Figure \ref{PSZadfig}, cf. \cite{Mi, Zad}, \cite[p.~192--193]{Bia G} or \cite[p.~228]{No}.

\begin{figure}[ht] 
\begin{center} 
\includegraphics[scale=0.17]{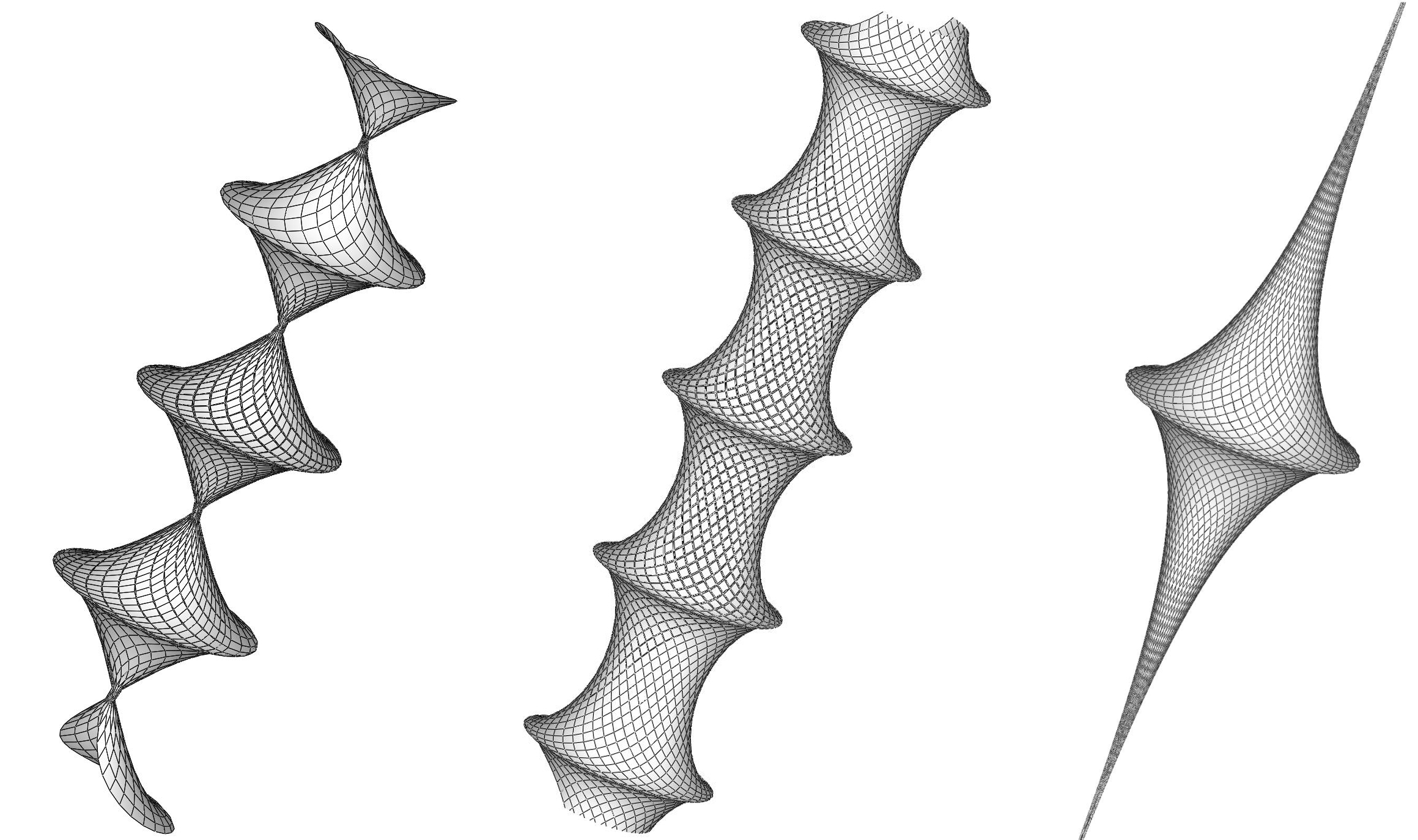}
\caption{From the left: Minding's pseudospherical surfaces $\mathbf{r}_0$ (parameterized by $\xi,\eta$) corresponding to solutions 
\eqref{sol1}, \eqref{sol2} and \eqref{sol3} respectively, $k = 1$, $l = 3/2$ (left), $l = 1/2$ (middle),  $l = 1$ (right). 
The rightmost surface is the pseudosphere.} \label{PSZadfig} 
\end{center}
\end{figure} 

Let us compute a B\"acklund transformation of $\mathbf r_0$.
Firstly, applying trigonometric identities, one obtains
$$
\sin(\frac{\omega_0 - \omega_0^{(\lambda)}}{2}) = \frac{-2\delta^{(\lambda)}}{1+{\delta^{(\lambda)}}^2} \cos\omega_0 
+ \frac{1 - {\delta^{(\lambda)}}^2}{1 + {\delta^{(\lambda)}}^2} \sin\omega_0, \\
\sin(\frac{\omega_0 + \omega_0^{(\lambda)}}{2}) = \frac{2\delta^{(\lambda)}}{1+{\delta^{(\lambda)}}^2}.
$$
Substituting into \eqref{BTsurf} yields
$$ \numbered\label{BTPSZad}
\mathbf r^{(\lambda)} = \mathbf r_0
 + \frac{2\lambda}{1 + \lambda^2} [ (\frac{1 - {\delta^{(\lambda)}}^2}{1 + {\delta^{(\lambda)}}^2} 
- 
 \frac{2\delta^{(\lambda)}\cot\omega_0}{1+{\delta^{(\lambda)}}^2} )  {\mathbf  r_{0,\xi} }
+
\frac{2\delta^{(\lambda)}\csc\omega_0}{1+{\delta^{(\lambda)}}^2}  {\mathbf  r_{0,\eta}}
],
$$
where
$$
\mathbf r_{0,\xi} = 
\frac{\sqrt{k}}{p(k+1) }
(
\begin{array}{c}
k \omega_{0,\alpha} \cos\frac{\omega_0}{2}  \sin[p (\xi-\eta)] + 2p \sin\frac{\omega_0}{2} \cos[p (\xi-\eta)]  \\
k \omega_{0,\alpha} \cos\frac{\omega_0}{2}  \cos[p (\xi-\eta)] - 2p \sin\frac{\omega_0}{2} \sin[p (\xi-\eta)] \\
1 + \sqrt{k} \cos \omega_0
\end{array}
),
$$
$$
\mathbf r_{0,\eta} = 
\frac{\sqrt{k}}{p(k+1) }
(
\begin{array}{c}
\omega_{0,\alpha} \cos\frac{\omega_0}{2}  \sin[p (\xi-\eta)]-2p \sin\frac{\omega_0}{2} \cos[p (\xi-\eta)]  \\
\omega_{0,\alpha} \cos\frac{\omega_0}{2}  \cos[p (\xi-\eta)]+2p \sin\frac{\omega_0}{2} \sin[p (\xi-\eta)]  \\
1 + \frac{\cos \omega_0}{\sqrt{k}}
\end{array}
)
$$
and $\omega_{0,\alpha}$ is given by \eqref{omegaalpha2}, i.e.
$$
\omega_{0,\alpha} = \pm\sqrt{\frac{2l - 2\cos\omega_0}{k}}.
$$

The transformed Minding's surfaces, $\mathbf r^{(\lambda)}$, 
corresponding to  cases A, B and C can be seen in Figure \ref{PStr}.  

\begin{figure}[ht] 
\begin{center} 
\includegraphics[scale=0.17]{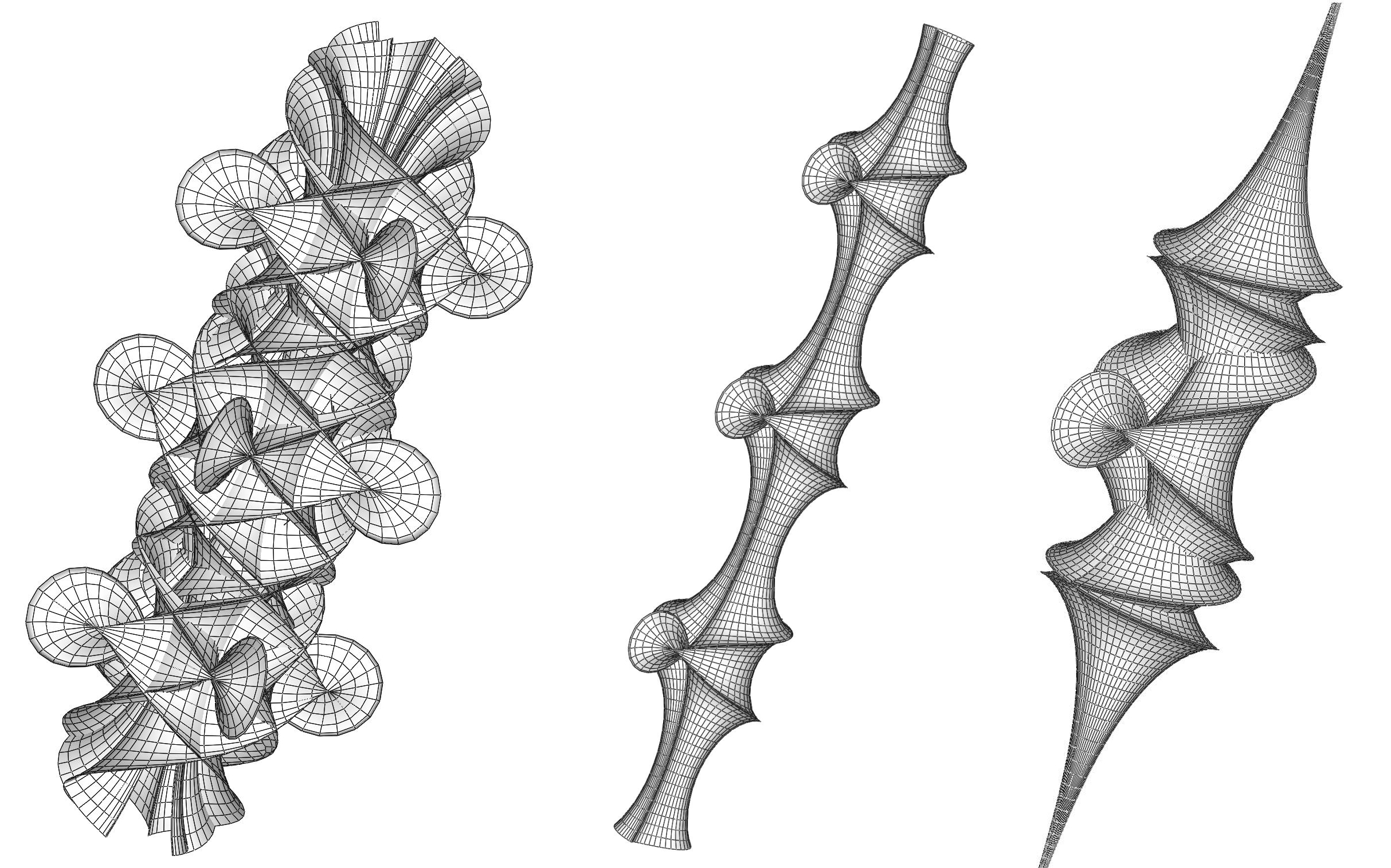}
\caption{Transformed Minding's pseudospherical surfaces (parameterized by $\alpha,\beta$) 
corresponding to cases A, B and C respectively, 
$k = 1, \lambda = 1, K = 0$, $l = 3/2$ (left), $l = 1/2$ (middle),  $l = 1$ (right). 
The righmost is the two-soliton Kuen's surface \cite{Ku}, see also \cite[p.~470]{Bia G}.} \label{PStr} 
\end{center}
\end{figure} 

Now we are at the point of constructing corresponding family of constant astigmatism surfaces, the common involutes of complementary pseudospherical surfaces $\mathbf r_0$ and $\mathbf{r}^{(1)}$.
The unit normal, $\tilde{\mathbf n}^{(1)}$, of the constant astigmatism surface is 
$$
\tilde{\mathbf n}^{(1)} = \mathbf r^{(1)} - \mathbf r_0 = 
 (\frac{1 - {\delta^{(1)}}^2}{1 + {\delta^{(1)}}^2} - \frac{2\delta^{(1)}\cot\omega_0}{1+{\delta^{(1)}}^2} )  {\mathbf  r_{0,\xi} }
+
\frac{2\delta^{(1)}\csc\omega_0}{1+{\delta^{(1)}}^2}  {\mathbf  r_{0,\eta}},
$$ 
where ${\delta^{(1)}}$ arises from \eqref{deltasol} by substituting $\lambda = 1$. 

Note that $\tilde{\mathbf n}^{(1)}(\xi,\eta)$ parameterizes the unit sphere by \it slip lines, \rm see Figure \ref{n1}. 
A~physical interpretation, see \cite{H-M I}, says that the
principal stress lines of a sphere deformed in tangential directions under the Tresca yield condition (e.g., a metallic sheet
obtaining a spherical shape) constitute so called \it orthogonal equiareal pattern \rm 
and the slip-lines are curves forming an angle of $\pi/4$ with principal stress lines. 

The  orthogonal equiareal pattern associated to $\tilde{\mathbf n}^{(1)}(\xi,\eta)$ is then  
$\tilde{\mathbf n}^{(1)}(x,y)$, where $x = x^{(1)}$ and $y = y^{(1)}$ are given by first two equations in \eqref{gx}. To perform the reparameterization, one has to express $\xi, \eta$ in terms of $x, y$. This is left as an open challenge. 

\begin{figure}[ht] 
\begin{center} 
\includegraphics[scale=0.15]{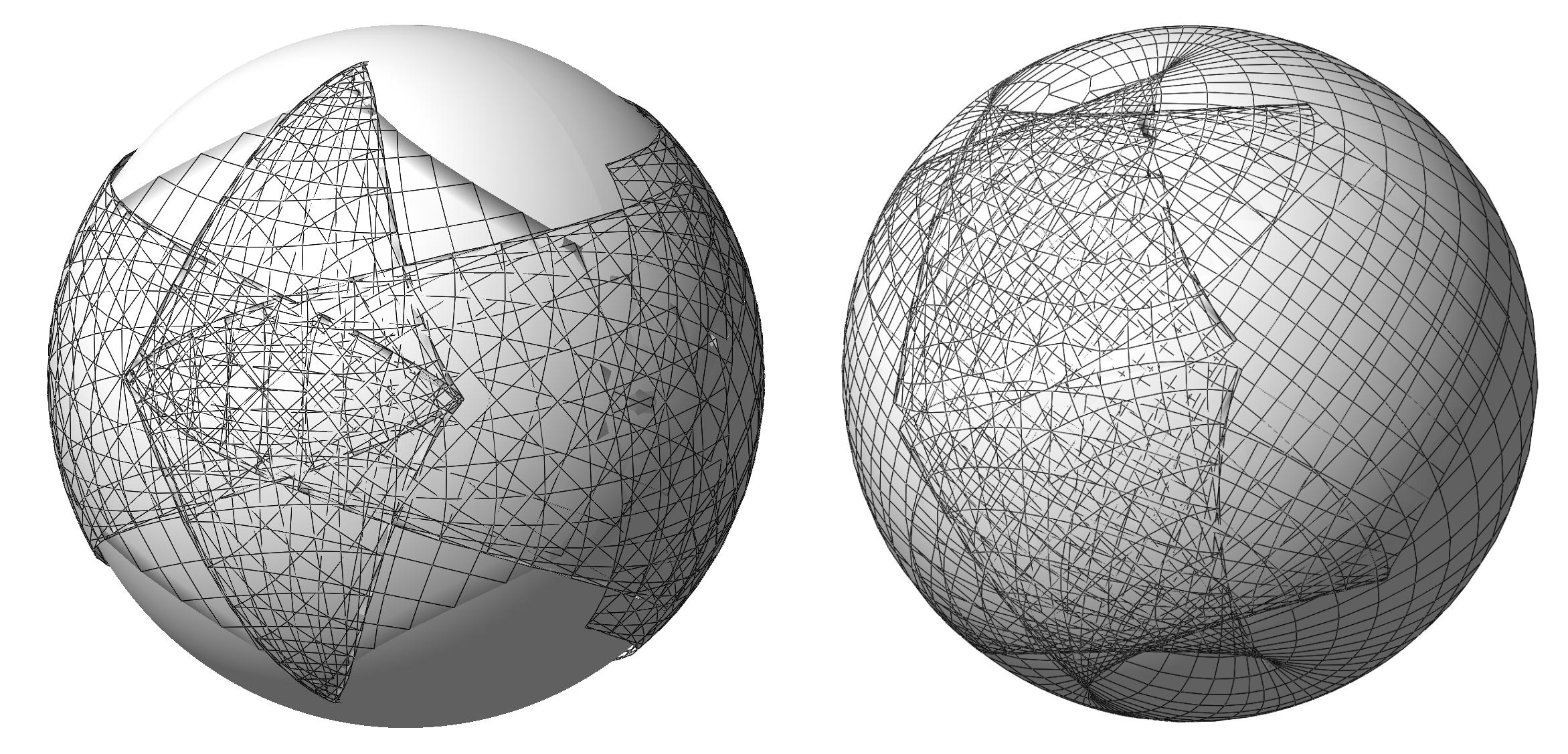}
\caption{Slip line fields $\tilde{\mathbf n}^{(1)}(\xi,\eta)$ corresponding to case A (left) and case B (right).} \label{n1} 
\end{center}
\end{figure} 

A family of constant astigmatism surfaces $\tilde{\mathbf r}^{(1)}$ is then given by \eqref{CAsurf}, see Figures \ref{CA1b} and \ref{CA1a}. 

\begin{figure}[ht] 
\begin{center} 
\includegraphics[scale=0.20]{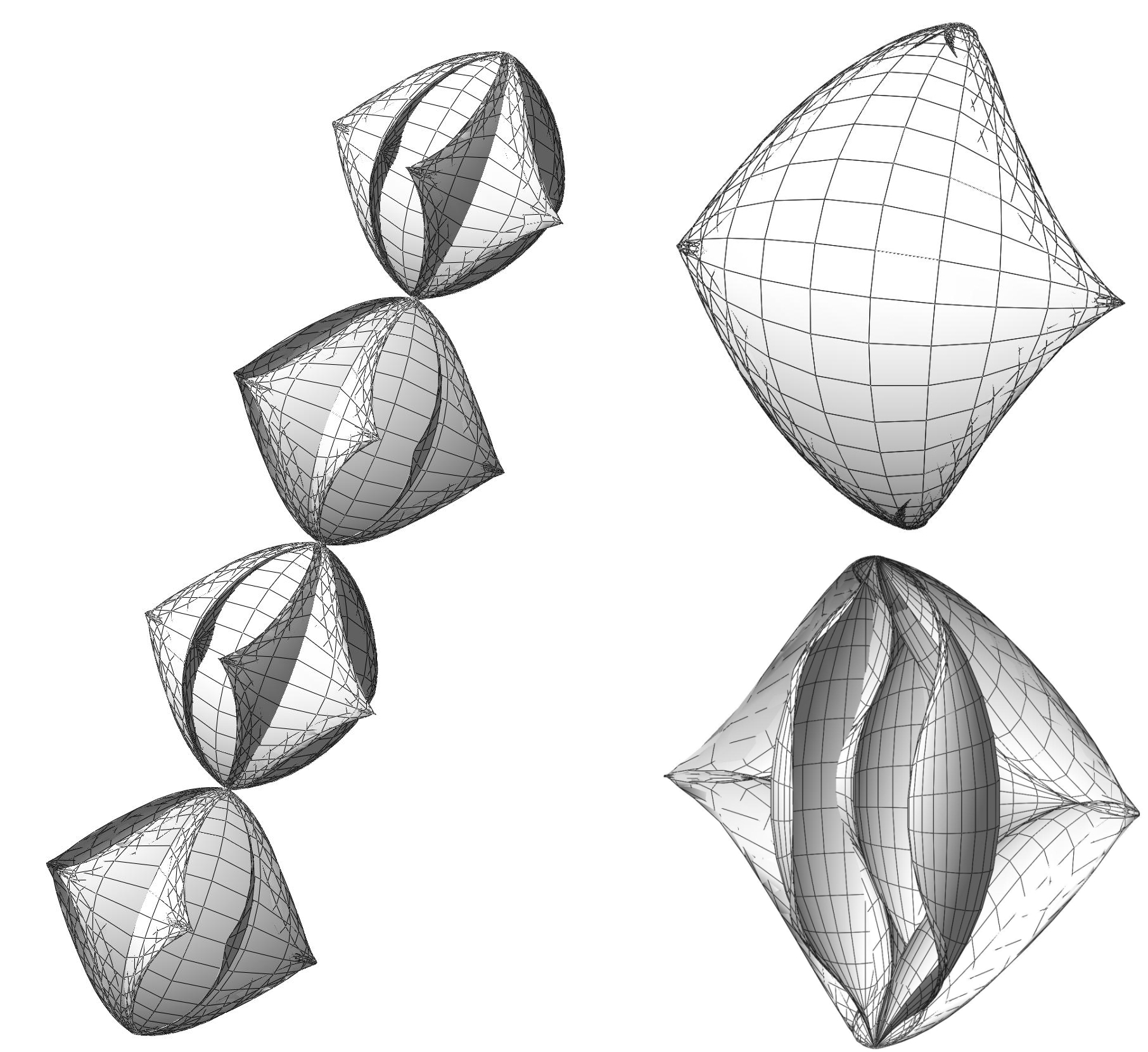}
\caption{Constant astigmatism surface $\tilde{\mathbf r}^{(1)}$, parameterized by $\alpha,\beta$,  corresponding to case~A. 
One of the pieces the surface is formed of is zoomed in the right. The piece is displayed from two mutually opposite directions.} \label{CA1b} 
\end{center}
\end{figure}

\begin{figure}[ht] 
\begin{center} 
\includegraphics[scale=0.13]{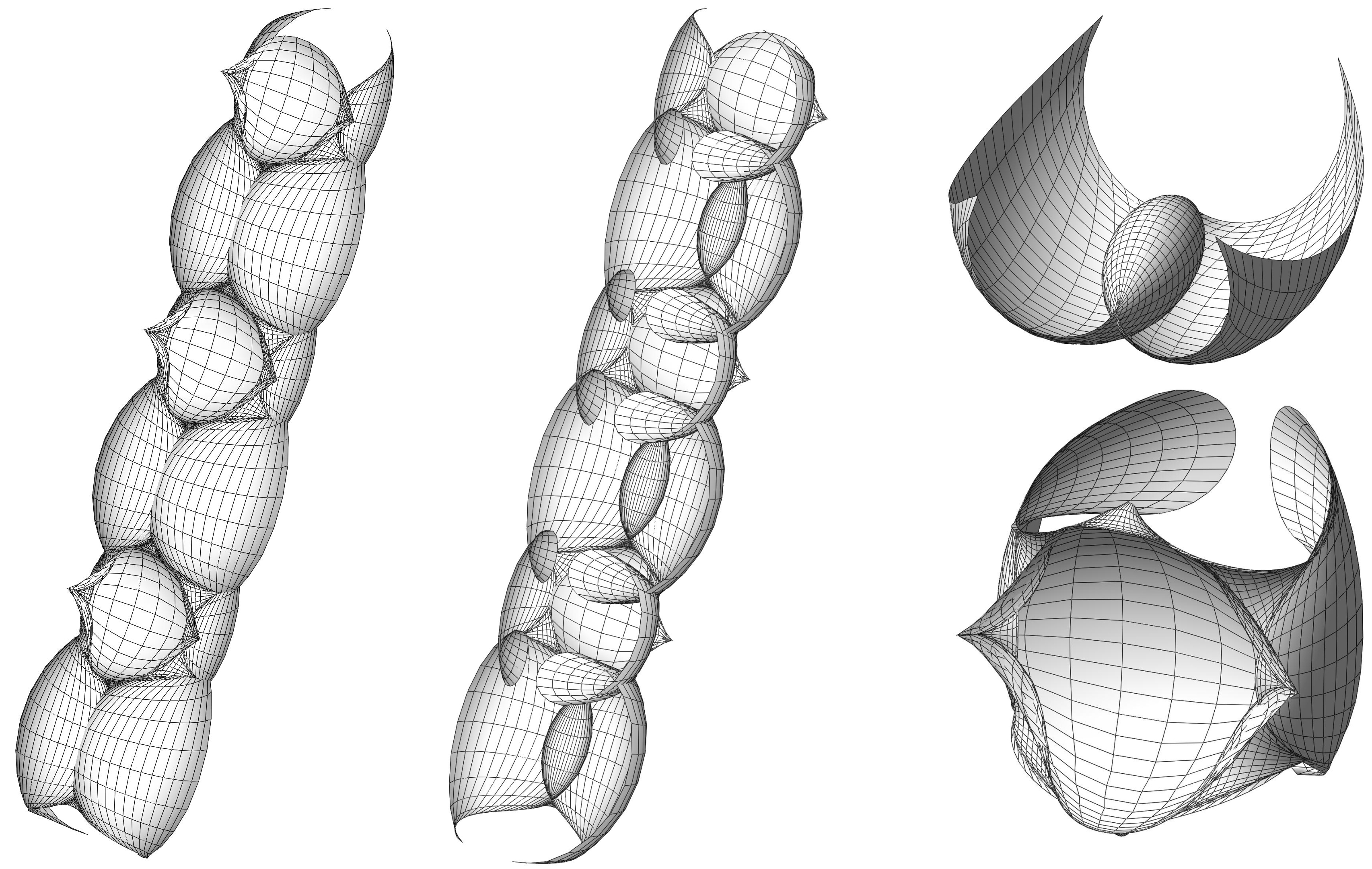}
\caption{Constant astigmatism surface $\tilde{\mathbf r}^{(1)}$, parameterized by $\alpha,\beta$, corresponding to case~B. 
The first two pictures from the left show two views of the same surface from mutually opposite directions. 
The repeated pieces are zoomed in the right.} \label{CA1a} 
\end{center}
\end{figure}

Let us proceed to one more example. Using Proposition \ref{prop1} we routinely construct solution $z^{(\lambda_1\lambda_2)}(x^{(\lambda_1\lambda_2)}, y^{(\lambda_1\lambda_2)})$ (see Figure \ref{CAsol2}) from known solution $z^{(\lambda_1)}(x^{(\lambda_1)}, y^{(\lambda_1)})$. Twice transformed Minding's pseudospherical surfaces $\mathbf{r}^{(\lambda_1,\lambda_2)}$ can be also routinely computed, see Figure \ref{PS12}. Corresponding 
constant astigmatism surfaces are plotted in Figure 
\ref{CA12}. 

\begin{figure}[ht] 
\begin{center} 
\includegraphics[scale=0.3]{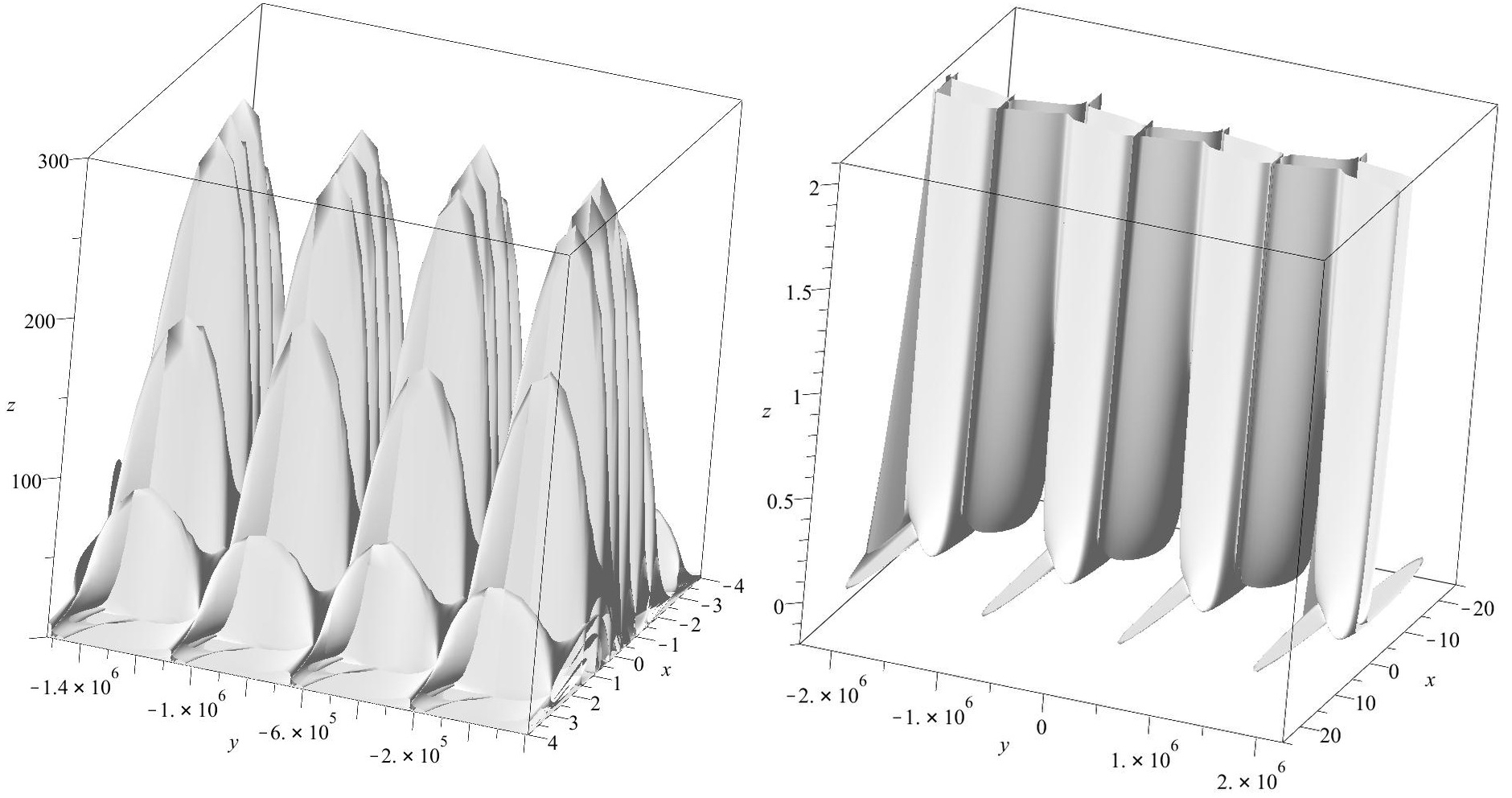}
\caption{A rather complicated multivalued solutions $z^{(\lambda_1\lambda_2)}(x^{(\lambda_1\lambda_2)}, y^{(\lambda_1\lambda_2)})$ of the CAE that are periodic in the $y$-direction, $\lambda_1 = 1.001$, $\lambda_2 = 1.3$, $k = 1$, $K = 0$, $l = 3/2$ (case A, left),  $l = 1/2$ (case B, right).} \label{CAsol2} 
\end{center}
\end{figure}

\begin{figure}[ht] 
\begin{center} 
\includegraphics[scale=0.24]{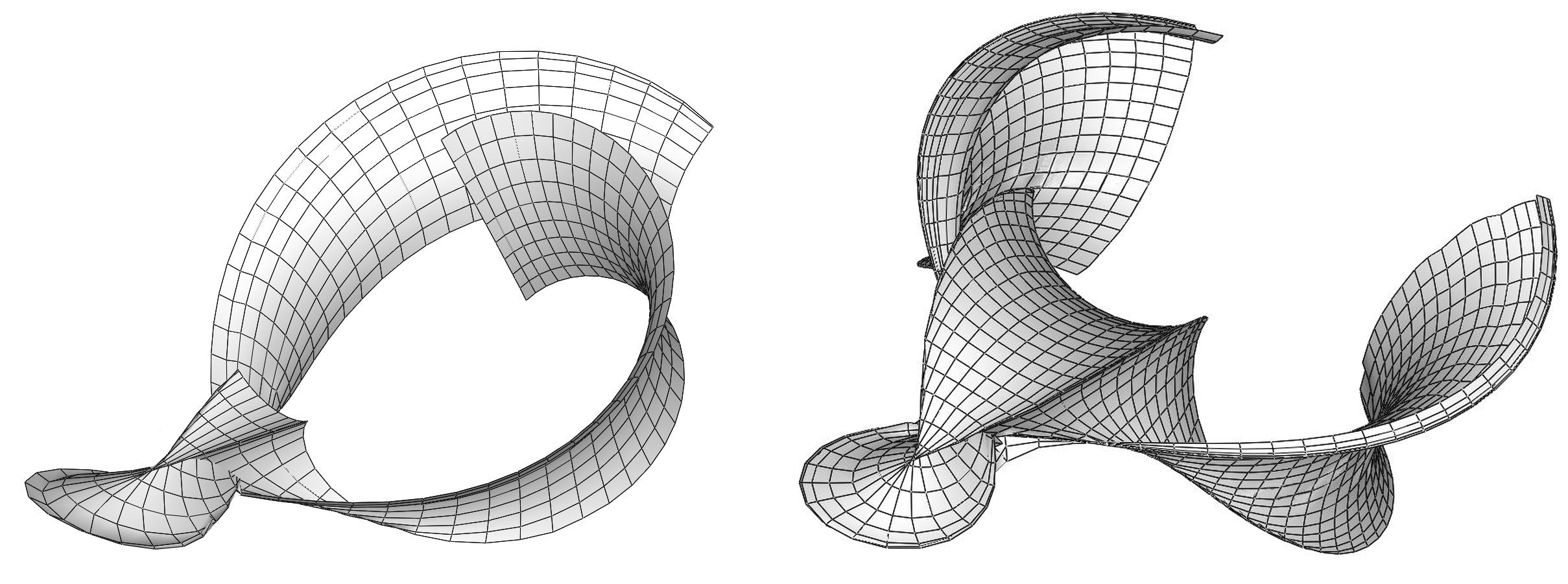}
\caption{Pieces of pseudospherical surfaces ${\mathbf r}^{(\lambda_1\lambda_2)}$, parameterized by 
$\alpha,\beta$, with $\lambda_1 = 1$, $\lambda_2 = 1.3$. Case A (left), case B (right).} \label{PS12} 
\end{center}
\end{figure}

\begin{figure}[ht] 
\begin{center} 
\includegraphics[scale=0.25]{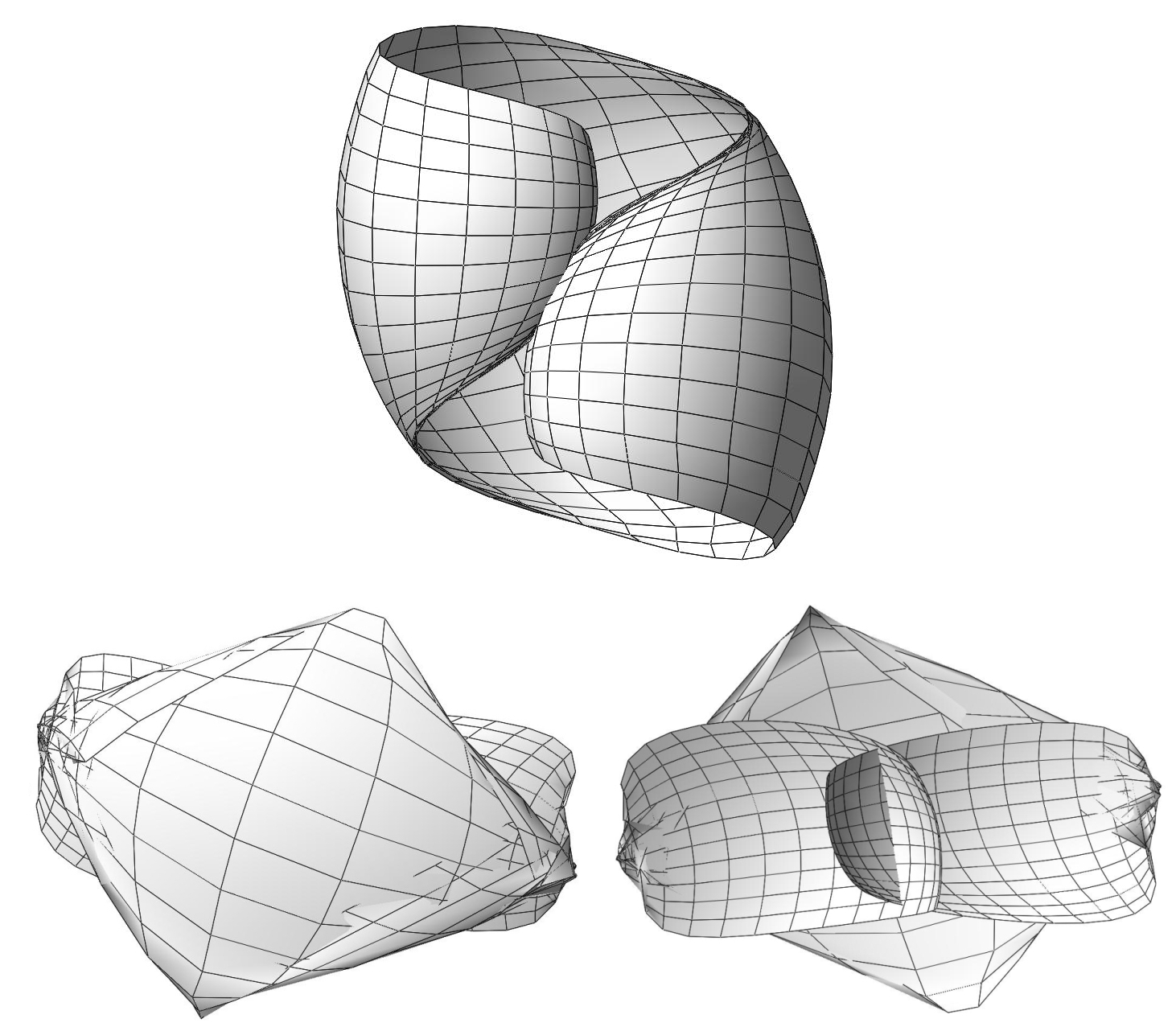}
\caption{Two pieces of constant astigmatism surface, the common involute of ${\mathbf r}^{(\lambda_2)}$ and ${\mathbf r}^{(1\lambda_2)}$ with $\lambda_2 = 1.3$, case B. The second piece is displayed from two views. Surface is parameterized by $\alpha,\beta$. }
\label{CA12}  
\end{center}
\end{figure}

\ack

The author was supported by Specific Research grant SGS/17/2016 of the Silesian University in Opava and
wishes to extend his gratitude to Michal Marvan and Petr Blaschke for valuable discussions.

\end{document}